\documentclass[abstract,headings=normal]{scrartcl}
\usepackage[english]{babel}
\usepackage[T1]{fontenc}
\usepackage[applemac]{inputenc}
\usepackage{epic,eepic,graphicx}
\usepackage{tikz}
\usepackage{amsmath}
\usepackage{amssymb}
\usepackage{amsthm}
\usepackage{subfig}
\usepackage{accents}
\usepackage{hyperref}

\def\setcircle#1{\def\rrr{#1}}
\def\Black(#1,#2){\put(#1,#2){\circle*{\rrr}}}
\def\White(#1,#2){\put(#1,#2){\circle{\rrr}}}

\newcommand{\bbC}{\mathbb{C}}
\newcommand{\bbZ}{\mathbb{Z}}
\newcommand{\bbR}{\mathbb{R}}

\newcommand{\cL}{\mathcal{L}}

\newcommand{\cX}{\mathcal{X}}

\newtheorem{theorem}{Theorem}
\newtheorem{proposition}[theorem]{Proposition}
\newtheorem{lemma}[theorem]{Lemma}
\newtheorem{corollary}[theorem]{Corollary}
\newtheorem{definition}[theorem]{Definition}
\theoremstyle{remark}

\begin{document}

\title{What is integrability of discrete variational systems?}
\author{Raphael Boll, Matteo Petrera, Yuri B. Suris}
\maketitle

\begin{center}
{
Institut f\"ur Mathematik, MA 7-2, Technische Universit\"at Berlin, \\
Str. des 17. Juni 136, 10623 Berlin, Germany\\
E-mail: {\tt boll, petrera, suris@math.tu-berlin.de}
}
\end{center}

\begin{abstract}
We propose a notion of a pluri-Lagrangian problem, which should be understood as an analog of multi-dimensional consistency for variational systems. This is a development along the line of research of discrete integrable Lagrangian systems initiated in 2009 by Lobb and Nijhoff, however having its more remote roots in the theory of pluriharmonic functions, in the Z-invariant models of statistical mechanics and their quasiclassical limit, as well as in the theory of variational symmetries going back to Noether. A $d$-dimensional pluri-Lagrangian problem can be described as follows: given a $d$-form $\cL$ on an $m$-dimensional space (called multi-time, $m>d$), whose coefficients depend on a sought-after function $x$ of $m$ independent variables (called field), find those fields $x$ which deliver critical points to the action functionals $S_\Sigma=\int_\Sigma\cL$ for {\em any} $d$-dimensional manifold $\Sigma$ in the multi-time. We derive the main building blocks of the multi-time Euler-Lagrange equations for a discrete pluri-Lagrangian problem with $d=2$, the so called corner equations, and discuss the notion of consistency of the system of corner equations. We analyze the system of corner equations for a special class of three-point 2-forms, corresponding to integrable quad-equations of the ABS list. This allows us to close a conceptual gap of the work by Lobb and Nijhoff by showing that the corresponding 2-forms are closed not only on solutions of (non-variational) quad-equations, but also on general solutions of the corresponding corner equations. We also find an example of a pluri-Lagrangian system not coming from a multidimensionally consistent system of quad-equations.
\end{abstract}

\newpage

\section{Introduction}

In the last decade, a new understanding of integrability of discrete systems as their multi-dimensional consistency has been a major breakthrough \cite{BS1}, \cite{N}. This led to classification of discrete 2-dimensional integrable systems (ABS list)  \cite{ABS}, which turned out to be rather influential. According to the concept of multi-dimensional consistency, integrable two-dimensional systems can be imposed in a consistent way on all two-dimensional sublattices of a lattice $\bbZ^m$ of arbitrary dimension. This means that the resulting multi-dimensional system possesses solutions whose restrictions to any two-dimensional sublattice are generic solutions of the corresponding two-dimensional system. To put this idea differently, one can impose the two-dimensional equations on any quad-surface in $\bbZ^m$ (i.e., a surface composed of elementary squares), and transfer solutions from one such surface to another one, if they are related by a sequence of local moves, each one involving one three-dimensional cube, like the moves shown of Fig. \ref{Fig: local moves}.

\begin{figure}[htbp]
\begin{center}
\includegraphics[width=0.5\textwidth]{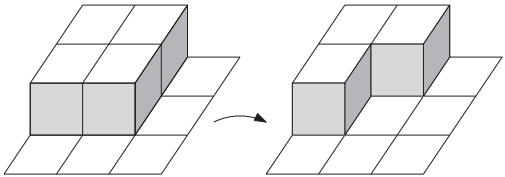}
\caption{Local move of a quad-surface involving one three-dimensional cube}
\label{Fig: local moves}
\end{center}
\end{figure}

A further fundamental conceptual development was initiated by Lobb and Nijhoff \cite{LN1} and deals with variational (Lagrangian) formulation of discrete multi-dimensionally consistent systems. Solutions of any ABS equation on any quad surface $\Sigma$ in $\bbZ^m$ are critical points of a certain action functional $S_\Sigma=\int_\Sigma\cL$ obtained by integration of a suitable discrete Lagrangian 2-form $\cL$. Lobb and Nijhoff observed that the value of the action functional remains invariant under local changes of the underlying quad-surface, and suggested to consider this as a defining feature of integrability. Their results, found on the case-by-case basis for some equations of the ABS list, have been extended to the whole list and were given a more conceptual proof in \cite{BS2} (generalized for asymmetric systems of quad-equations in \cite{BollSuris2}). Subsequently, this research was generalized in various directions: for multi-field two-dimensional systems \cite{LN2}, for dKP, the fundamental three-dimensional discrete integrable system \cite{LNQ}, and for the discrete time Calogero-Moser system, an important one-dimensional integrable system \cite{YLN}. Based on the latter example, a general theory of multi-time one-dimensional Lagrangian systems has been developed by one of us \cite{S}. It was demonstrated that, both in the continuous and in the discrete time, the property that the solutions deliver critical points for actions $S_\Gamma=\int_\Gamma \cL$ along arbitrary curves $\Gamma$ in the multi-time, is characteristic for commutativity of Hamiltonian flows in the continuous situation and of symplectic maps in the discrete situation. This property yields that the exterior derivative of the multi-time Lagrangian 1-form is constant, $d\cL={\rm const}$. Vanishing of this constant (i.e., closedness of the Lagrangian 1-form) was shown to be related to integrability in the following, more strict, sense:
\begin{itemize}
\item[--] in the continuous time situation, $d\cL=0$ is equivalent for the Hamiltonians of the commuting flows to be in involution,

\item[--] in the discrete time situation, for one-parameter families of commuting symplectic maps, $d\cL=0$ is equivalent to the {\em spectrality property}, introduced by Kuznetsov and Sklyanin \cite{KS}, which says that the derivative of the Lagrangian with respect to the parameter of the family is a generating function of common integrals of motion for the whole family.
\end{itemize}
An application of this general theory to integrable Toda-type systems and their B\"acklund transformations was given in our recent paper \cite{BoMaSu}.
\smallskip

Thus, the original idea of Lobb and Nijhoff can be summarized as follows: solutions of integrable systems deliver critical points simultaneously for actions along all possible manifolds of the corresponding dimension in multi-time; the Lagrangian form is closed on solutions. This idea is, doubtless, rather inventive (not to say exotic) in the framework of the classical calculus of variations. However,  it has significant precursors. These are:
\begin{itemize}
\item Theory of pluriharmonic functions and, more generally, of pluriharmonic maps \cite{R, OV, BFPP}. By definition, a pluriharmonic function of several complex variables $f:\bbC^m\to\bbR$ minimizes the Dirichlet functional $E_\Gamma=\int_\Gamma |(f\circ \Gamma)_z|^2dz\wedge d\bar z$ along any holomorphic curve in its domain $\Gamma:\bbC\to\bbC^m$. Differential equations governing pluriharmonic functions (and maps) are heavily overdetermined. Therefore it is not surprising that they belong to the theory of integrable systems.
\item Baxter's Z-invariance of solvable models of statistical mechanics  \cite{Bax1, Bax2}. This concept is based on invariance of the partition function of solvable models under elementary local transformations of the underlying planar graph. It is well known (see, e.g., \cite{BoMeSu}) that one can associate the planar graphs underling these models with quad-surfaces in $\bbZ^m$. On the other hand, the classical mechanical analogue of the partition function is the action functional. This makes the relation of Z-invariance to the concept of closedness of the Lagrangian 2-form rather natural, at least at the heuristic level. Moreover, this relation has been made mathematically precise for a number of models, through the quasiclassical limit, in the work of Bazhanov, Mangazeev, and Sergeev \cite{BMS1, BMS2}.
\item The classical notion of variational symmetries,  going back to the seminal work of E.~Noether \cite{Noether}, turns out to be directly related to the idea of the closedness of the Lagrangian form in the multi-time. We plan to elucidate this in a forthcoming publication.
\end{itemize}

Especially the relation with the pluriharmonic functions motivates us to introduce a novel term to describe the situation we are interested in, namely: given a $d$-form $\cL$ in the $m$-dimensional space ($d<m$), depending on a sought-after function $u$ of $m$ variables, one looks for functions $u$ which deliver critical points to actions $S_\Sigma=\int_\Sigma \cL$ corresponding to any $d$-dimensional manifold $\Sigma$. We call this a  {\em pluri-Lagrangian problem} and claim that integrability of variational systems should be understood as the existence of the pluri-Lagrangian structure. We find this term much more concrete and suggestive than the neutral one ``Lagrangian multiform structure'' used in the pioneering publications by Lobb and Nijhoff with co-authors. We envisage that this notion will play a very important role in the future development of the theory of integrable systems.

However, apart from giving a precise definition and a detailed discussion of discrete pluri-Lagrangian systems, which is done in Section \ref{sect: pluri}, the concrete goals of the present paper are more modest. We provide the reader with a more concrete motivation for the present study in Section \ref{sect: motiv}. Then, in Section \ref{sect: pluri quad} we consider pluri-Lagrangian systems which serve as a truly variational generalization of integrable quad-equations. Finally, in Section \ref{sect: pluri non-quad} we show that there exist discrete pluri-Lagrangian systems which cannot be derived from quad-equations.

\section{Motivation}
\label{sect: motiv}

The motivation for the present study comes from the the desire to answer the following two long-standing questions.
\smallskip

{\em Problem 1. Do discrete Laplace-type systems type exhibit some multi-dimensionally consistency?}
\noindent

Recall that an important instance of integrable Laplace-type systems are discrete relativistic Toda-type systems
\[
g(\widetilde{x}_k-x_k)-g(x_k-\undertilde{x}_k)=f(x_{k+1}-x_k)-f(x_k-x_{k-1})+
h(\undertilde{x}_{k+1}-x_k)-h(x_k-\widetilde{x}_{k-1}).
\]
Here, the tilde denotes the shift in the discrete time, which is the second coordinate direction on the lattice $\bbZ^2$, the first one being represented by the index $k$. This equation is illustrated in Fig.~\ref{fig: triangular lattice}. See \cite{A, AS, BollSuris1} for details. Systems of this kind can be considered as variational systems on the regular triangular lattice coming from the following action functional:
\begin{equation}\label{eq: action RTL}
    S=\sum_{\bbZ^2}\Big(L_1(x,x_1)+L_2(x,x_2)+L_3(x_1,x_2)\Big).
\end{equation}
Here we switched to the notation which will be mainly used throughout the paper: $x_1$, $x_2$ denote the shifts of the field $x$ in the 1st, resp. the 2nd coordinate direction of $\bbZ^2$.

\setcircle{10}
\begin{figure}[htbp]
\begin{center}
\setlength{\unitlength}{0.05em}
\begin{picture}(450,270)(-25,-30)
 \path(-25,200)(325,200)
 \multiput(0,0)(100,0){4}{\Black(0,200)}
 \multiput(0,0)(100,0){4}{\Black(100,0)}
 \multiput(0,100)(100,0){5}{\Black(0,0)}
 \path(0,100)(0,200)\path(100,100)(100,200)
 {\allinethickness{0.75mm} \path(200,100)(200,200)}
 \path (300,100)(300,200)
 \path(100,0)(100,100)
 {\allinethickness{0.75mm} \path(200,0)(200,100)}
 \path(300,0)(300,100)\path(400,0)(400,100)
 \path(0,200)(100,100)
 {\allinethickness{0.75mm} \path(100,200)(200,100)}
 \path(200,200)(300,100)\path(300,200)(400,100)
 \path(0,100)(100,0)\path(100,100)(200,0)
 {\allinethickness{0.75mm} \path(200,100)(300,0)}
 \path(300,100)(400,0)
 \path(-25,100)(100,100)
 {\allinethickness{0.75mm} \path(100,100)(300,100)}
 \path(300,100)(425,100)
 \path(-25,125)(0,100)
 \path(400,100)(400,125)
 \path(0,75)(0,100) \path(75,0)(425,0) \path(400,100)(425,75)
 \multiputlist(100,210)(100,0)[cb]{$\widetilde{x}_{k-1}$,$\widetilde{x}_k$}
 \multiputlist(110,110)(100,0)[lb]{$x_{k-1}$,$x_k$,$x_{k+1}$}
 \multiputlist(200,-10)(100,0)[ct]{$\undertilde{x}_k$,$\undertilde{x}_{k+1}$}
\end{picture}
\caption{Regular triangular lattice underlying discrete relativistic Toda-type systems; one-index notation is used, the shift in the second coordinate direction (discrete time) being denoted by tilde}
\label{fig: triangular lattice}
\end{center}
\end{figure}
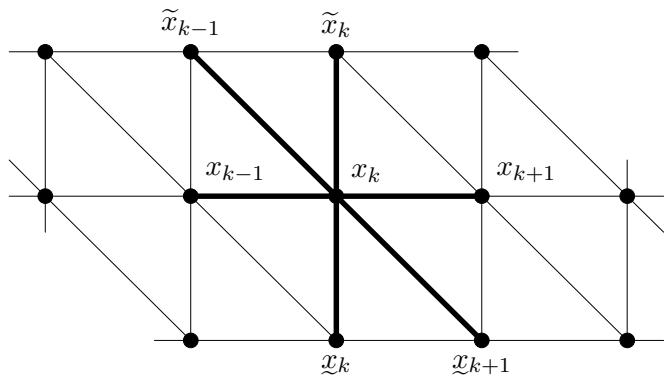

Despite a rather complete understanding of relations between integrable Laplace-type systems and integrable (multi-dimensionally consistent) systems of quad-equations (cf. \cite{BS1, AS, BollSuris1}), the question whether Laplace-type systems themselves possess any property of the kind of multi-dimensional consistency, resisted our understanding for quite a long time. Already the very first step in answering this question is far from obvious: what should be the combinatorial structures supporting a multi-dimensional extension of a discrete Laplace type system, for instance, of a system on the regular triangular lattice of Fig.~\ref{fig: triangular lattice}?
\medskip

{\em Problem 2. What are the natural conditions for the validity of the closure relation for action in multi-dimensions?}

\noindent
Recall that the main idea of the paper \cite{LN1} was to extend the functional
(\ref{eq: action RTL}) to quad-surfaces $\Sigma$ in the multidimensional square lattice according to the formula
$$
S_\Sigma=\sum_{\sigma_{ij}\in\Sigma}\cL(\sigma_{ij}),
$$
where the discrete Lagrangian 2-form $\cL$ assigns to each oriented elementary square $\sigma_{ij}=(n,n+e_i,n+e_i+e_j,n+e_j)$ the number $\cL(\sigma_{ij})=\cL_{ij}(x,x_i,x_j)$.
The main result, discovered in \cite{LN1} for a part of the ABS list and proven in \cite{BS2} for the whole ABS list and in \cite{BollSuris2} for asymmetric systems of quad-equations, states that for any consistent system of quad-equations on $\bbZ^m$, the Lagrangian 2-form $\cL$ satisfies the following closure relation {\em on solutions of quad-equations}:
\begin{equation*}\label{closure relation}
\Delta_i\mathcal{L}(\sigma_{jk})+\Delta_j\mathcal{L}(\sigma_{ki})+\Delta_k\mathcal{L}(\sigma_{ij})=0.
\end{equation*}
One feature of this important result is, however, not completely satisfying. The action $S_\Sigma$ provides only a {\em weak Lagrangian formulation} of quad-equations. This means that Euler-Lagrange equations for the action functional $S_\Sigma$ are mere {\em consequences} of quad equations. One cannot derive quad-equations from the Euler-Lagrange equations for $S_\Sigma$. In other words, solutions of quad-equations constitute a rather small subset of solutions of the Euler-Lagrange equations for the action functional $S_\Sigma$. Thus, it would be conceptually appealing to prove the closure relation {\em on solutions of the Euler-Lagrange equations} for an arbitrary surface $\Sigma$. For this, again the very first step should consist in a convenient description of the full set of the variational equations for all quad-surfaces $\Sigma$.
\smallskip

A solution for both problems mentioned here is provided in the next section and consists in the introduction of the concept of pluri-Lagrangian problems (Definition \ref{def:pluriLagr problem}) and in the finding a compact form of the full system of variational equations for such problems (Definition \ref{def:pluriLagr system}).

\section{Pluri-Lagrangian systems}
\label{sect: pluri}

\begin{definition}\label{def:pluriLagr problem} {\bf (Pluri-Lagrangian problem)}

\begin{itemize}
\item Let $\cL$ be a discrete 2-form, i.e., a real-valued function of oriented elementary squares
\[
\sigma_{ij}=\left(n,n+e_{i},n+e_{i}+e_{j},n+e_{j}\right)
\]
of $\bbZ^m$, such that $\cL\left(\sigma_{ij}\right)=-\cL\left(\sigma_{ji}\right)$. Usually, we will assume that it depends on some field assigned to the points of $\bbZ^m$, that is, on some $x:\bbZ^m\to \cX$ ($\cX$ being some vector space).
\item To an arbitrary oriented quad-surface $\Sigma$ in $\bbZ^m$ with a topology of a disk, there corresponds the {\em action functional}, which assigns to $x|_{V(\Sigma)}$, i.e., to the fields at the vertices of the surface $\Sigma$, the number
\begin{equation}\label{action on surface}
S_\Sigma=\sum_{\sigma_{ij}\in\Sigma}\cL(\sigma_{ij}).
\end{equation}
\item We say that the field $x:V(\Sigma)\to \cX$ is a critical point of $S_\Sigma$, if at any interior point $n\in V(\Sigma)$, we have
\begin{equation}\label{eq: dEL gen}
    \frac{\partial S_\Sigma}{\partial x(n)}=0.
\end{equation}
Equations (\ref{eq: dEL gen}) are called {\em discrete Euler-Lagrange equations} for the action $S_\Sigma$ (or just for the surface $\Sigma$, if it is clear which 2-form $\cL$ we are speaking about).
\item We say that the field $x:\bbZ^m\to\cX$ solves the {\em pluri-Lagrangian problem} for the Lagrangian 2-form $\cL$ if, {\em for any quad-surface $\Sigma$ in $\bbZ^m$}, the restriction $x|_{V(\Sigma)}$ is a critical point of the corresponding action $S_\Sigma$.
\end{itemize}
\end{definition}

Since the combinatorics of a quad-surface $\Sigma$ around its generic interior point $n$ can be rather complicated, it might seem that a general study of discrete Euler-Lagrange equations is a rather hopeless enterprize. We will show, however, that this is not the case, namely, that all Euler-Lagrange equations can be built from some elementary building blocks.
\begin{definition} {\bf (3D-corner)}
A 3D-corner is a quad-surface consisting of three elementary squares adjacent to a vertex of valence 3.
\end{definition}

This definition is illustrated in Fig.~\ref{fig:corners}.

\begin{figure}[htbp]
   \centering
   \subfloat[3D-corner at $n$]{\label{fig:1a}
   \setlength{\unitlength}{0.04em}
\begin{picture}(220,280)(0,-40)
 \put(0,0){\circle*{10}}    \put(150,0){\circle*{10}}
 \put(0,150){\circle*{10}}   \put(150,150){\circle*{10}}
 \put(50,200){\circle*{10}}
 \put(50,50){\circle*{10}}   \put(200,50){\circle*{10}}
 \path(150,0)(150,150)
 \path(150,0)(200,50)
 \path(0,0)(50,50)
 \path(50,50)(200,50)
 \path(150,150)(0,150)
 \path(200,50)(150,0)
 \path(0,145)(0,5)
 \path(0,0)(150,0)
 \path(0,150)(50,200)
 \path(50,200)(50,50)
 \put(0,-25){$n$}
 \put(7,130){$n+e_k$}
 \put(140,-25){$n+e_i$}
 \put(43,25){$n+e_j$}
%
\end{picture}}\qquad
   \subfloat[3D-corner at $n+e_i$]{\label{fig:1b}
      \setlength{\unitlength}{0.04em}
\begin{picture}(220,280)(0,-40)
 \put(0,0){\circle*{10}}    \put(150,0){\circle*{10}}
 \put(0,150){\circle*{10}}   \put(150,150){\circle*{10}}
 \put(200,200){\circle*{10}}
 \put(50,50){\circle*{10}}   \put(200,50){\circle*{10}}
 \path(150,0)(150,150)
 \path(150,0)(200,50)
 \path(0,0)(50,50)
 \path(50,50)(200,50)
 \path(150,150)(0,150)
 \path(200,50)(150,0)
 \path(0,145)(0,5)
 \path(0,0)(150,0)
 \path(150,150)(200,200)
 \path(200,200)(200,50)
 \put(0,-25){$n$}
 \put(140,-25){$n+e_i$}
 \put(160,140){$n+e_i+e_k$}
 \put(160,60){$n+e_i+e_j$}
%
\end{picture}}\qquad
   \subfloat[3D-corner at $n+e_i+e_j$]{\label{fig:1c}
      \setlength{\unitlength}{0.04em}
\begin{picture}(220,280)(0,-40)
 \put(0,0){\circle*{10}}    \put(150,0){\circle*{10}}
 \put(150,150){\circle*{10}}
 \put(50,200){\circle*{10}} \put(200,200){\circle*{10}}
 \put(50,50){\circle*{10}}   \put(200,50){\circle*{10}}
 \path(150,0)(200,50)
 \path(55,50)(200,50)
 \path(0,0)(50,50)
 \path(0,0)(150,0)
 \path(150,0)(150,150)
  \path(50,50)(50,200)
 \path(150,150)(200,200)
 \path(50,200)(200,200)
 \path(200,200)(200,50)
%
 \put(140,-25){$n+e_i$}
 \put(160,60){$n+e_i+e_j$}
 \put(100,215){$n+e_i+e_j+e_k$}
 \put(45,25){$n+e_j$}
%
\end{picture}
   }
   \caption{Three 3D-corners of an elementary 3D cube}
   \label{fig:corners}
\end{figure}
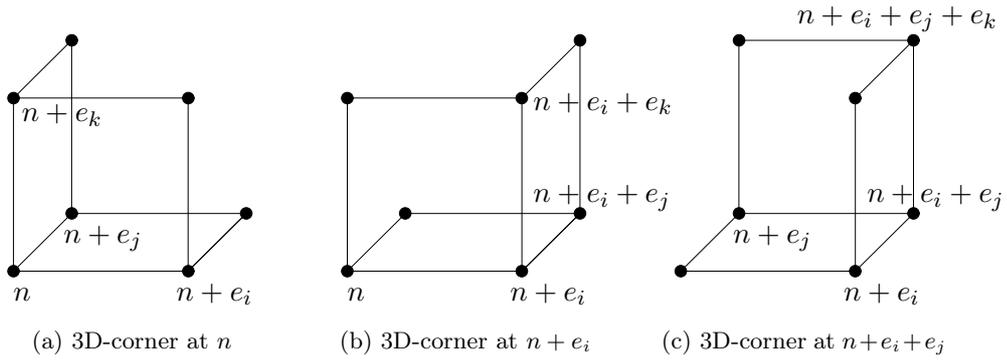

\begin{lemma}
The flower of any interior vertex of an oriented quad-surface in $\bbZ^m$ can be represented as a sum of (oriented) 3D-corners in $\bbZ^{m+1}$.
\end{lemma}
\begin{proof}
Consider the flower of an interior vertex $n$ of $\Sigma$. Over each elementary square (petal) of the flower, we can build a 3D-corner spanned by the two edges of this petal adjacent to $n$ and the edge $(n,n+e_{m+1})$ of an additional coordinate direction. The orientation of this 3D-corner is induced by the orientation of the corresponding elementary square of $\Sigma$. Then, obviously, the ``vertical'' squares coming from two successive petals of the flower carry opposite orientations, so that all ``vertical'' squares cancel away from the sum of the oriented 3D-corners. The construction is illustrated by the case of a four-petal flower in Fig.~\ref{fig:flower}.
\end{proof}

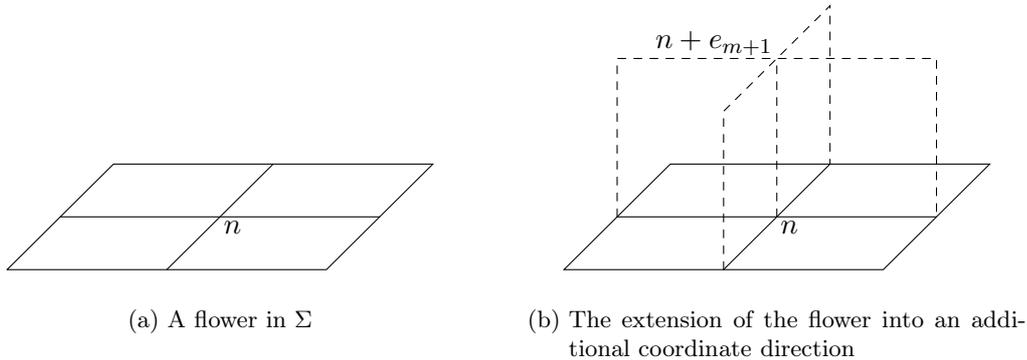
\begin{figure}[htbp]
   \centering
   \subfloat[A flower in $\Sigma$]{\label{fig:flower1}
   \begin{tikzpicture}[scale=0.7,inner sep=0.5]
      \useasboundingbox (-0.5,-0.5) rectangle (8.5,2.5);
      \draw (0,0) to (3,0) to (6,0) to (7,1) to (8,2) to (5,2) to (2,2) to (1,1) to (0,0);
      \draw (1,1) to (4,1) to (7,1);
      \draw (3,0) to (4,1) to (5,2);
      \node (n) at (4,1) [label=-45:$n$]{};
   \end{tikzpicture}
   }\qquad
   \subfloat[The extension of the flower into an additional coordinate direction]{\label{fig:flower2}
   \begin{tikzpicture}[scale=0.7,inner sep=0.5]
      \useasboundingbox (-0.5,-0.5) rectangle (8.5,5.5);
      \draw (0,0) to (3,0) to (6,0) to (7,1) to (8,2) to (5,2) to (2,2) to (1,1) to (0,0);
      \draw (1,1) to (4,1) to (7,1);
      \draw (3,0) to (4,1) to (5,2);
      \draw [dashed] (1,1) to (1,4) to (4,4) to (7,4) to (7,1);
      \draw [dashed] (3,0) to (3,3) to (4,4) to (5,5) to (5,2);
      \draw [dashed] (4,1) to (4,4);
      \node (n) at (4,1) [label=-45:$n$]{};
      \node (n) at (4,4) [label=135:$n+e_{m+1}$]{};
   \end{tikzpicture}
   }
   \caption{A planar four-petal flower and its decomposition into a sum of four 3D-corners}
   \label{fig:flower}
\end{figure}

As a consequence, the action for any flower can be represented as a sum of actions for several 3D-corners. Thus, the Euler-Lagrange equation for any interior vertex $n$ of $\Sigma$ can be decomposed into a sum of several Euler-Lagrange equations for 3D-corners. This justifies the following fundamental definition.

\begin{definition}\label{def:pluriLagr system} {\bf (System of corner equations)} The {\em system of corner equations} for a given discrete 2-form $\cL$ consists of discrete Euler-Lagrange equations for all possible 3D-corners in $\bbZ^m$. If the action for the surface of an oriented elementary cube $\sigma_{ijk}$ of the coordinate directions $i,j,k$ (which can be identified with the discrete exterior derivative $d\cL$ evaluated at $\sigma_{ijk}$) is denoted by
\begin{equation}\label{eq: Sijk}
S^{ijk}=d\cL(\sigma_{ijk})=\Delta_k\cL(\sigma_{ij})+\Delta_i\cL(\sigma_{jk})+\Delta_j\cL(\sigma_{ki}),
\end{equation}
then the system of corner equations consists of the eight equations
\begin{equation}\label{eq: corner eqs}
\begin{array}{llll}
\dfrac{\partial S^{ijk}}{\partial x}=0, & \dfrac{\partial S^{ijk}}{\partial x_i}=0, & \dfrac{\partial S^{ijk}}{\partial x_j}=0, & \dfrac{\partial S^{ijk}}{\partial x_k}=0, \\
\\
\dfrac{\partial S^{ijk}}{\partial x_{ij}}=0, & \dfrac{\partial S^{ijk}}{\partial x_{jk}}=0, & \dfrac{\partial S^{ijk}}{\partial x_{ik}}=0, & \dfrac{\partial S^{ijk}}{\partial x_{ijk}}=0
\end{array}
\end{equation}
for each triple $i,j,k$. Symbolically, this can be put as $\delta(d\cL)=0$, where $\delta$ stands for the ``vertical'' differential (differential with respect to the dependent field variables $x$).
\end{definition}
Thus, the system of corner equations encompasses all possible discrete Euler-Lagrange equations for all possible quad-surfaces $\Sigma$. In other words, solutions of a pluri-Lagrangian problem as introduced in Definition \ref{def:pluriLagr problem} are precisely solutions of the corresponding system of corner equations.
\medskip

{\bf Remark.} We formulated the system of corner equations for a generic 2-form $\cL$. In particular cases the quantity $S^{ijk}$ could be independent on some of the fields at the corners of the cube. Then the system of corner equations (\ref{eq: corner eqs}) could contain less equations. An example of such a situation will be studied in detail in the next section.
\medskip

Of course, in order that the above definition be meaningful, the system of corner equations has to be {\em consistent}. Indeed, the system of corner equations (\ref{eq: corner eqs}) for one elementary cube is heavily overdetermined. It consists of eight equations, each one connecting seven fields out of eight. Any six fields can serve as independent data, then one can use two of the corner equations to compute the remaining two fields, and the remaining six corner equations have to be satisfied identically. This justifies the following definition.

\begin{definition}\label{def: corner eqs consist} {\bf (Consistency of corner equations)} System (\ref{eq: corner eqs}) is called {\em consistent}, if it has the minimal possible rank 2, i.e., if exactly two of these equations are independent.
\end{definition}

The main feature of our definition is that the ``almost closedness'' of the 2-form $\cL$ on solutions of the system of corner equations is, so to say, built-in from the outset. One should compare the proof of the following theorem with similar proofs in \cite{BS2, S}.

\begin{theorem}\label{Th: almost closed}
For any triple of the coordinate directions $i,j,k$, the action $S^{ijk}$ over an elementary cube of these coordinate directions is constant on solutions of the system of corner equations (\ref{eq: corner eqs}):
\[
S^{ijk}(x,\ldots,x_{ijk})=c^{ijk}={\rm const}  \pmod{\partial S^{ijk}/\partial x=0,\ \ldots,\
\partial S^{ijk}/\partial x_{ijk}=0}.
\]
\end{theorem}
\begin{proof} On the connected six-dimensional manifold of solutions, the gradient of $S^{ijk}$ considered as a function of eight variables, vanishes by virtue of (\ref{eq: corner eqs}).
\end{proof}
The most interesting case is, of course, when all $c^{ijk}=0$. Then one can say that $d\cL=0$, i.e., the discrete 2-form $\cL$ is closed on solutions of the system of corner equations.

\section{System of corner equations for a three-point 2-form corresponding to integrable quad-equations}
\label{sect: pluri quad}

The main class of examples we consider in this paper is characterized by the following ansatz for the discrete 2-form $\cL$:
\begin{equation}\label{lag LN ij}
\cL(\sigma_{ij})=\mathcal{L}(x,x_i,x_j;\alpha_i,\alpha_j)=L(x,x_i;\alpha_i)-
L(x,x_j;\alpha_j)-\Lambda(x_i,x_j;\alpha_i,\alpha_j)
\end{equation}
for each elementary square $\sigma_{ij}=(n,n+e_i,n+e_i+e_j,n+e_j)$.
Of course, the function $\Lambda$ should satisfy $\Lambda(x,y;\alpha,\beta)=-\Lambda(y,x;\beta,\alpha)$ to ensure the skew-symmetry property $\cL(\sigma_{ji})=-\cL(\sigma_{ij})$. Note that this is the general structure of the discrete 2-form used in \cite{LN1, BS2} to describe the Lagrangian structure of integrable quad-equations of the ABS list. With minor modifications, results of the present section can be extended to asymmetric systems of quad-equations studied in \cite{BollSuris2}. We stress again that quad-equations {\em are not variational}; rather, discrete Euler-Lagrange equations for the 2-forms given in \cite{LN1, BS2} are {\em consequences} of quad-equations. The main result of \cite{LN1, BS2} says that the discrete 2-forms $\cL$ are closed on solutions of the corresponding quad-equations.
We now turn to the intrinsically Lagrangian generalization of this statement.

First of all, we mention that for a three-point 2-form (\ref{lag LN ij}), the action $S^{ijk}$ over the surface of an elementary cube specializes to
\begin{eqnarray}\label{eq: dL}
S^{ijk} & = & L(x_i,x_{ij};\alpha_j)+L(x_j,x_{jk};\alpha_k)+L(x_k,x_{ik};\alpha_i)\nonumber\\
       &   & -L(x_i,x_{ik};\alpha_k)-L(x_j,x_{ij};\alpha_i)-L(x_k,x_{jk};\alpha_j)\nonumber\\
       &   & -\Lambda(x_{ij},x_{ik};\alpha_j,\alpha_k)
             -\Lambda(x_{jk},x_{ij};\alpha_k,\alpha_i)
             -\Lambda(x_{ik},x_{jk};\alpha_i,\alpha_j)\nonumber\\
       &   & +\Lambda(x_{j},x_{k};\alpha_j,\alpha_k)+
             \Lambda(x_{k},x_{i};\alpha_k,\alpha_i)+
             \Lambda(x_{i},x_{j};\alpha_i,\alpha_j).\qquad
\end{eqnarray}
Thus, $S^{ijk}$ depends on neither $x$ nor $x_{ijk}$, so that its  domain of definition is better visualized as an octahedron shown in Fig.~\ref{octahedron}.

\begin{figure}[htbp]
\begin{center}
\setlength{\unitlength}{0.05em}
\begin{picture}(200,250)(0,-10)
 \put(0,0){\circle{10}}    \put(150,0){\circle*{10}}
 \put(0,150){\circle*{10}}   \put(150,150){\circle*{10}}
 \put(50,200){\circle*{10}} \put(200,200){\circle{10}}
 \put(50,50){\circle*{10}}   \put(200,50){\circle*{10}}
 {\allinethickness{0.65mm}\path(150,5)(150,150)(5,150)
 \path(3.53,153.53)(50,200)(150,150)(200,50)(153.53,3.53)
 \path(3.53,146.47)(146.47,3.53)
 \path(1,145.5)(46.5,54)
 \path(145.5,1)(54,46.5)
 \dashline[100]{15}(50,200)(200,50)
 \dashline[80]{15}(50,55)(50,200)
 \dashline[80]{15}(55,50)(200,50)}
 {\thinlines\path(3.53,3.53)(46.47,46.47)
 \path(0,145)(0,5)\path(5,0)(145,0)
 \path(150,150)(196.47,196.47)
 \path(50,200)(195,200)
 \path(200,195)(200,50)}
 \put(-20,-15){$x$} \put(-30,145){$x_k$}
 \put(160,-15){$x_i$} \put(162,140){$x_{ik}$}
 \put(215,45){$x_{ij}$} \put(210,210){$x_{ijk}$}
 \put(40,25){$x_j$}  \put(35,213){$x_{jk}$}
\end{picture}
\caption{Octahedron}\label{octahedron}
\end{center}
\end{figure}
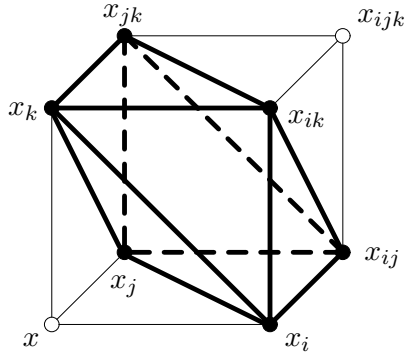

The system of corner equations for $\cL$ consists of six equations for each elementary cube: for the vertices $x$ and $x_{ijk}$ there are {\em no equations}, while for the vertices $x_i$ and $x_{ij}$ we have the following {\em four-leg, five-point equations}:
\begin{equation}\label{eq: corner i}\tag{$E_i$}
    \psi(x_i,x_{ij};\alpha_j)-\psi(x_i,x_{ik};\alpha_k)-
    \phi(x_{i},x_{k};\alpha_i,\alpha_k)+\phi(x_{i},x_{j};\alpha_i,\alpha_j)=0,
\end{equation}
and
\begin{equation}\label{eq: corner ij}\tag{$E_{ij}$}
    \psi(x_{ij},x_{i};\alpha_j)-\psi(x_{ij},x_{j};\alpha_i)-
    \phi(x_{ij},x_{ik};\alpha_j,\alpha_k)+\phi(x_{ij},x_{jk};\alpha_i,\alpha_k)=0.
\end{equation}
Here, we introduced the notation
\begin{equation*}\label{eq: psi}
\psi(x,y;\alpha)=\frac{\partial L(x,y;\alpha)}{\partial x},\quad
\phi(x,y;\alpha,\beta)=\frac{\partial \Lambda(x,y;\alpha,\beta)}{\partial x}.
\end{equation*}
These equations are illustrated in Fig.~\ref{fig: corner eqs}.

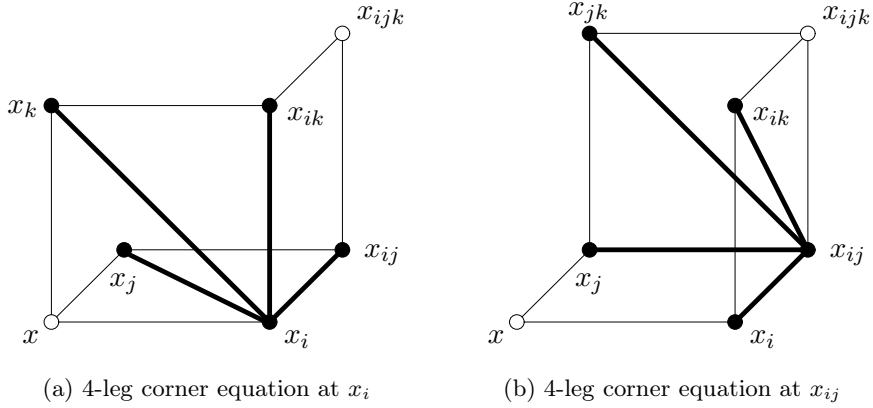
\begin{figure}[htbp]
   \centering
   \subfloat[4-leg corner equation at $x_{i}$]{\label{fig:2c}
   \setlength{\unitlength}{0.05em}
\begin{picture}(220,280)(0,-30)
 \put(0,0){\circle{10}}    \put(150,0){\circle*{10}}
 \put(0,150){\circle*{10}}   \put(150,150){\circle*{10}}
 \put(200,200){\circle{10}}
 \put(50,50){\circle*{10}}   \put(200,50){\circle*{10}}
 {\allinethickness{0.65mm}\path(150,5)(150,150)
 \path(3.53,146.47)(146.47,3.53)
 \path(145.5,1)(54,46.5)
 \path(150,0)(200,50)}
 {\thinlines\path(3.53,3.53)(46.47,46.47)
 \path(50,50)(200,50)
 \path(150,150)(5,150)
 \path(200,50)(153.53,3.53)
 \path(0,145)(0,5)\path(5,0)(145,0)
 \path(150,150)(196.47,196.47)
 \path(200,195)(200,50)}
 \put(-20,-15){$x$} \put(-30,145){$x_k$}
 \put(160,-15){$x_i$} \put(162,140){$x_{ik}$}
 \put(215,45){$x_{ij}$} \put(210,210){$x_{ijk}$}
 \put(40,25){$x_j$}  
\end{picture}
   }\qquad\qquad
   \subfloat[4-leg corner equation at $x_{ij}$]{\label{fig:2d}
   \setlength{\unitlength}{0.05em}
\begin{picture}(220,280)(0,-30)
 \put(0,0){\circle{10}}    \put(150,0){\circle*{10}}
 \put(150,150){\circle*{10}}
 \put(50,200){\circle*{10}} \put(200,200){\circle{10}}
 \put(50,50){\circle*{10}}   \put(200,50){\circle*{10}}
 {\allinethickness{0.65mm}
 \path(150,150)(200,50)
 \path(150,0)(200,50)
 \path(50,200)(200,50)
 \path(55,50)(200,50)}
 {\thinlines\path(3.53,3.53)(46.47,46.47)
 \path(5,0)(150,0)
 \path(150,5)(150,150)
  \path(50,55)(50,200)
 \path(150,150)(196.47,196.47)
 \path(50,200)(195,200)
 \path(200,195)(200,50)}
 \put(-20,-15){$x$} 
 \put(160,-15){$x_i$} \put(162,140){$x_{ik}$}
 \put(215,45){$x_{ij}$} \put(210,210){$x_{ijk}$}
 \put(40,25){$x_j$}  \put(35,213){$x_{jk}$}
\end{picture}
   }
   \caption{Two sorts of corner equations in the case of a three-point 2-form}
   \label{fig: corner eqs}
\end{figure}

The notion of {\em consistency} of the system of corner equations has to be modified in the present case as follows. One can take any four of the six fields as independent initial data. Then two of the corner equations can be used to determine the two remaining fields, and then the four remaining corner equations have to be satisfied identically. We see that formally Definition \ref{def: corner eqs consist} can be kept literally the same.
\medskip

In Fig.~\ref{fig: from corner eqs to 7-point}, we translate geometric considerations of Fig.~\ref{fig:flower} to the level of equations, showing how the sum of four corner equations for a three-point 2-form yields a planar equation of the relativistic Toda type (cf. also Fig.~\ref{fig: triangular lattice}).

\begin{figure}[htbp]
   \centering

   \setlength{\unitlength}{0.037em}
\begin{picture}(1000,500)(0,-20)
%
%
 \put(0,0){\circle{10}}
 \put(150,0){\circle*{10}}
 \put(150,150){\circle*{10}}
 \put(50,200){\circle*{10}}
 \put(200,200){\circle{10}}
 \put(50,50){\circle*{10}}
 \put(200,50){\circle*{10}} \put(200,25){$x$}
 {\allinethickness{0.65mm}
 \path(150,150)(200,50)
 \path(150,0)(200,50)
 \path(50,200)(200,50)
 \path(55,50)(200,50)}
 {\thinlines\path(3.53,3.53)(46.47,46.47)
 \path(5,0)(150,0)
 \path(150,5)(150,150)
  \path(50,55)(50,200)
 \path(150,150)(196.47,196.47)
 \path(50,200)(195,200)
 \path(200,195)(200,50)}
%
%
%
 \put(100,250){\circle{10}}
 \put(250,250){\circle*{10}} \put(260,240){$x$}
 \put(100,400){\circle*{10}}
 \put(250,400){\circle*{10}}
 \put(300,450){\circle{10}}
 \put(150,300){\circle*{10}}
 \put(300,300){\circle*{10}}
 {\allinethickness{0.65mm}\path(250,255)(250,400)
 \path(103.53,396.47)(246.47,253.53)
 \path(245.5,251)(154,296.5)
 \path(250,250)(300,300)}
 {\thinlines\path(103.53,253.53)(146.47,296.47)
 \path(150,300)(300,300)
 \path(250,400)(105,400)
 \path(300,300)(253.53,253.53)
 \path(100,395)(100,255)\path(105,250)(245,250)
 \path(250,400)(296.47,446.47)
 \path(300,445)(300,300)}
%
%
%
 \put(250,0){\circle{10}}
 \put(400,0){\circle*{10}}
 \put(250,150){\circle*{10}}
 \put(300,200){\circle*{10}}
 \put(450,200){\circle{10}}
 \put(300,50){\circle*{10}}  \put(295,25){$x$}
 \put(450,50){\circle*{10}}
 {\allinethickness{0.65mm}
 \path(300,50)(400,0)
 \path(300,50)(450,50)
 \path(300,50)(300,200)
 \path(300,50)(250,150)}
 {\thinlines\path(253.53,3.53)(300,50)
 \path(255,0)(400,0)
 \path(250,5)(250,150)(300,200)(450,200)(450,50)(400,0)
 }
%
%
 \put(350,250){\circle{10}}  \put(325,240){$x$}
 \put(500,250){\circle*{10}}
 \put(350,400){\circle*{10}}
 \put(500,400){\circle*{10}}
 \put(400,450){\circle*{10}}
 \put(400,300){\circle*{10}}
 \put(550,300){\circle*{10}}
 {\thinlines\path(353.53,253.53)(400,300)
 \path(355,250)(500,250)
 \path(350,255)(350,400)(500,400)(500,250)(550,300)(400,300)(400,450)(350,400)
 }
\put(250,200){$\boldsymbol{+}$}
%
%
\put(550,150){\circle{10}}
\put(700,150){\circle*{10}}
\put(850,150){\circle*{10}}
\put(600,200){\circle*{10}}
\put(750,200){\circle*{10}} \put(745,175){$x$}
\put(900,200){\circle*{10}}
\put(650,250){\circle*{10}}
\put(800,250){\circle*{10}}
\put(950,250){\circle{10}}
 {\thinlines\path(553.53,153.53)(650,250)(945,250)
 \path(946.47,246.47)(850,150)(553,150)
 }
{\allinethickness{0.65mm}
 \path(600,200)(900,200)
 \path(700,150)(800,250)
 \path(650,250)(850,150)}
\put(545,200){$\boldsymbol{=}$}
\end{picture}
   \caption{Sum of four corner equations (to be matched at $x$) results in a planar seven-point equation of the relativistic Toda type}
   \label{fig: from corner eqs to 7-point}
\end{figure}
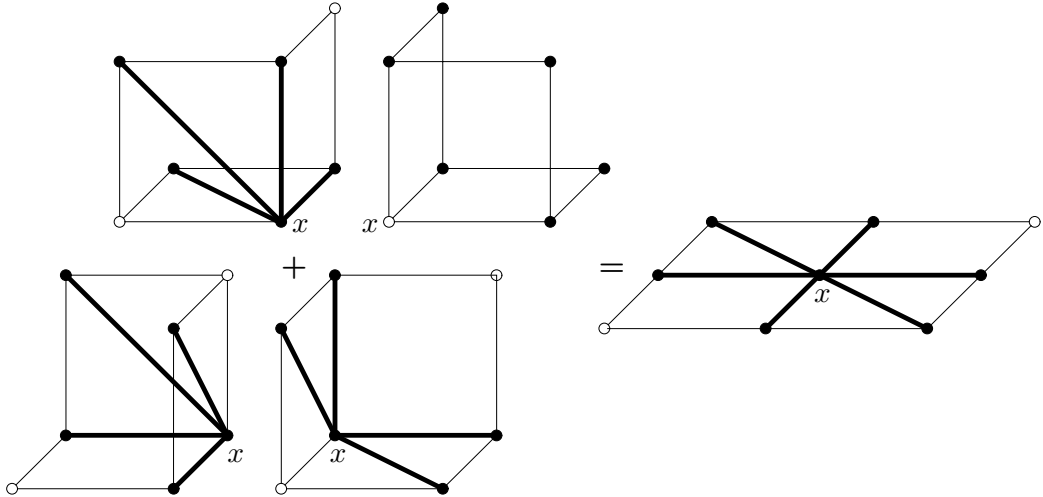

Now we can take action functionals given in \cite{LN1, BS2} for consistent systems of quad-equations and consider the corresponding systems of corner equations. We remark that the corner equations within one cube do not know anything about the fields $x$ and $x_{ijk}$. In particular, quad-equations do not have to be satisfied (no equation of the type (\ref{eq: corner i}) or (\ref{eq: corner ij}) contains 4 vertices of an elementary square). Effectively, this means that in \cite{LN1, BS2} one considered these same systems, but only a very restricted class of solutions thereof, satisfying quad-equations on all squares. Here, we are interested in general solutions.

To phrase this on the more concrete level of equations, we observe that each of the corner equations (\ref{eq: corner i}), (\ref{eq: corner ij}) is a difference of the corresponding three-leg forms of two quad-equations. For instance, equation (\ref{eq: corner i}) is a difference of
\begin{eqnarray}
\psi(x_i,x_{ij};\alpha_j)+\phi(x_i,x_j;\alpha_i,\alpha_j)-\psi(x_i,x;\alpha_i)=0,
  \label{eq: 3leg ij}\\
\psi(x_i,x_{ik};\alpha_k)+\phi(x_i,x_k;\alpha_i,\alpha_k)-\psi(x_i,x;\alpha_i)=0,
  \label{eq: 3leg ik}
\end{eqnarray}
which are three-leg forms, centered at $x_i$, of the quad-equations on the elementary squares $(x,x_i,x_{ij},x_j)$, resp. $(x,x_i,x_{ik},x_k)$. Now we do not require that these three-leg equations hold, but only that their four-leg difference is satisfied.

\begin{theorem} \label{Th: pluri quad consist}
For the discrete 2-forms $\cL$ given in \cite{LN1, BS2} for multi-dimensionally consistent quad-equations of the ABS list, the corresponding systems of corner equations are consistent, as well. Moreover, the 2-form $\cL$ is closed on solutions of the corner equations.
\end{theorem}
\begin{proof}
Consider initial data at four out of six vertices of an elementary 3D cube with one or two indices. For the sake of concreteness, let them be the values $x_i$, $x_j$, $x_{ij}$, $x_{ik}$. Find the other two by using two of the corner equations. In our example these should be equations $(E_i)$ delivering $x_k$ and $(E_{ij})$ delivering $x_{jk}$. Now we {\em define} an auxiliary field $x$ by requiring that the quad-equation on the face $(x,x_i,x_{ij},x_j)$ be fulfilled. In other words, equation (\ref{eq: 3leg ij}) is used to define $x$. Note that this value $x$ is ``alien'' in the sense that it has nothing to do neither  with the solution of the system of corner equations (which does not contain the corner equation at the vertex $x$ corresponding to the 3D cube under consideration) nor with the solution of discrete Euler-Lagrange equations for any 2D surface, even if it contains the vertex $x$. This ``alien'' value of $x$ satisfies also the quad-equation on the face $(x,x_i,x_{ik},x_k)$ in its three-leg form (\ref{eq: 3leg ik}), as it follows from the comparison of the four-leg corner equation (\ref{eq: corner i}) with the three-leg equation (\ref{eq: 3leg ij}). Now, the ``alien'' value $x$ together with the values $x_i$, $x_j$, $x_k$, yields, due to the 3D consistency of the system of quad-equations, a solution of the latter on the 3D cube. We show that the restriction of this solution to the vertices with one or two indices satisfies {\em all} corner equations of our system. Indeed, from the three-leg forms of the quad-equations on the faces adjacent to $x$, centered at $x_j$, $x_k$, we see that the corner equations $(E_j)$, $(E_k)$ are satisfied, as well. Analogously, from the three-leg forms of the quad-equations on the faces adjacent to $x_{ijk}$, centered at the vertices with two indices, we see that all three corner equations at the vertices with two indices are also fulfilled.

As for the second claim (closedness of the form $\cL$), it follows from Theorem \ref{Th: almost closed} that the quantities $S^{ijk}$ are constant on solutions of the corner equations. The values of the constants $c^{ijk}$ can be determined on {\em any} solution. For solutions which satisfy the corresponding quad-equations, these constants are equal to zero by results of \cite{LN1, BS2}.
\end{proof}

We will show that consistency of some of the systems of corner equations of Theorem \ref{Th: pluri quad consist} has another spectacular manifestation. For this, we start with the following observation.

\begin{proposition}
Solutions of quad-equations of the types Q1, Q3$_{\delta=0}$, H1, H2, H3 on an elementary 3D cube satisfy {\em octahedron relations}, i.e., relations of the type $$Q(x_i,x_j,x_k,x_{ij},x_{jk},x_{ik})=0,$$ which do not involve $x$ and $x_{ijk}$.
\end{proposition}
\begin{proof}
By inspection of the following list. Octahedron relations are obtained by adding/ multiplying three equations adjacent to $x$ in the form given in the list.
\medskip

\begin{description}
\item[Q1$_{\delta=0}$,] cross-ratio equation:
\[
\frac{(x-x_i)(x_{ij}-x_j)}{(x_i-x_{ij})(x_j-x)}=\frac{\alpha_i}{\alpha_j},
\]

 octahedron relation:
\[
\frac{(x_{ij}-x_i)(x_{jk}-x_j)(x_{ki}-x_k)}{(x_{ij}-x_j)(x_{jk}-x_k)(x_{ki}-x_i)}=1.
\]

\item[Q1$_{\delta=1}$,] shifted cross-ratio equation:
\[
\frac{(x-x_i+\alpha_i)(x_{ij}-x_j+\alpha_i)}{(x_i-x_{ij}-\alpha_j)(x_j-x-\alpha_j)}=
\frac{\alpha_i}{\alpha_j},
\]

 octahedron relation:
\[
\frac{(x_{ij}-x_i+\alpha_j)(x_{jk}-x_j+\alpha_k)(x_{ki}-x_k+\alpha_i)}
{(x_{ij}-x_j+\alpha_i)(x_{jk}-x_k+\alpha_j)(x_{ki}-x_i+\alpha_k)}=1.
\]

\item[Q3$_{\delta=0}$,] hyperbolic shifted cross-ratio equation:
\[
\frac{\sinh(x-x_i+\alpha_i)\sinh(x_{ij}-x_j+\alpha_i)}
{\sinh(x_i-x_{ij}-\alpha_j)\sinh(x_j-x-\alpha_j)}=
\frac{\sinh(2\alpha_i)}{\sinh(2\alpha_j)},
\]

octahedron relation:
\[
\frac{\sinh(x_{ij}-x_i+\alpha_j)\sinh(x_{jk}-x_j+\alpha_k)\sinh(x_{ki}-x_k+\alpha_i)}
{\sinh(x_{ij}-x_j+\alpha_i)\sinh(x_{jk}-x_k+\alpha_j)\sinh(x_{ki}-x_i+\alpha_k)}=1.
\]

\item[H1,] discrete KdV equation:
\[
(x-x_{ij})(x_i-x_j)=\alpha_i-\alpha_j,
\]

octahedron relation:
\[
x_{ij}(x_i-x_j)+x_{jk}(x_j-x_k)+x_{ki}(x_k-x_i)=0.
\]

\item[H2]:
\[
(x-x_{ij})(x_i-x_j)+(\alpha_j-\alpha_i)(x+x_i+x_j+x_{ij})+\alpha_j^2-\alpha_i^2=0,
\]

octahedron relation:
\[
x_{ij}(x_i-x_j+\alpha_i-\alpha_j)+x_{jk}(x_j-x_k+\alpha_j-\alpha_k)+x_{ki}(x_k-x_i+\alpha_k-\alpha_i)
\]
\[
+x_i(\alpha_k-\alpha_j)+x_j(\alpha_i-\alpha_k)+x_k(\alpha_j-\alpha_i)=0.
\]

\item[H3,] Hirota equation:
\[
\alpha_i(xx_i+x_jx_{ij})-\alpha_j(xx_j+x_ix_{ij})+\delta(\alpha_i^2-\alpha_j^2)=0,
\]

octahedron relation:
\[
\alpha_ix_jx_{ij}-\alpha_jx_ix_{ij}+\alpha_jx_kx_{jk}-\alpha_kx_jx_{jk}+
\alpha_kx_ix_{ki}-\alpha_ix_kx_{ki}=0.
\]
\end{description}
Observe that all these octahedron relations are (up to simple transformations of dependent variables) instances of the list of integrable 3D octahedron equations classified in \cite{ABS4}. However, not all items of the list show up in the present context, namely, only $(\chi_2)$ and $(\chi_3)$.
\end{proof}
\begin{corollary}
For the discrete 2-forms $\cL$ given in \cite{LN1, BS2} for quad-equations of the types Q1, Q3$_{\delta=0}$, H1, H2, H3, solutions of the corresponding systems of corner equations on an elementary 3D cube satisfy the respective octahedron relations.
\end{corollary}

But actually, a much more detailed statement can be made. It provides us also with an alternative proof of consistency of the corresponding systems of corner equations.
\begin{theorem} \label{th: consistency from octahedron}
For the discrete 2-forms $\cL$ given in \cite{LN1, BS2} for quad-equations of the types Q1, Q3$_{\delta=0}$, H1, H2, H3:
\begin{itemize}
\item if any two of the corner equations are satisfied, then the respective octahedron relation is satisfied, as well;
\item if any one of the corner equations and the octahedron relation are satisfied, then all other five corner equations are satisfied, as well.
\end{itemize}
In other words, all six corner equations are equivalent modulo the octahedron relation.
\end{theorem}
\begin{proof}
By a direct case-by-case computation. Let us describe the computation needed to prove the equivalence of corner equations $(E_i)$ and $(E_j)$, say, modulo the octahedron relation. Equation $(E_i)$ can be solved for $x_{ik}$ in terms of the four fields $x_i$, $x_j$, $x_k$, $x_{ij}$. Substitution of this expression into the octahedron relation yields an equation which relates the five fields $x_i$, $x_j$, $x_k$, $x_{ij}$, and $x_{jk}$. One has to check that this equation is algebraically equivalent to $(E_j)$. We will illustrate this with the simplest example of the system of corner equations with
\[
L(x,y;\alpha)=\alpha\log|x-y|, \quad \Lambda(x,y;\alpha,\beta)=(\alpha-\beta)\log|x-y|,
\]
corresponding to the cross-ratio equation Q1$_{\delta=0}$. Corner equation
$(E_i)$ is given by
\begin{equation*}\label{eq: Q1 Ei}
 \frac{\alpha_j}{x_{ij}-x_i}-\frac{\alpha_k}{x_{ik}-x_i}-\frac{\alpha_j-\alpha_i}{x_j-x_i}+
 \frac{\alpha_k-\alpha_i}{x_k-x_i} = 0.
\end{equation*}
To express $x_{ik}$ through the remaining variables, one transforms this equation to
\begin{equation*}\label{eq: Q1 Ei trans}
 \frac{\alpha_k(x_{ik}-x_k)}{(x_{ik}-x_i)(x_k-x_i)}  =  \frac{\alpha_i(x_j-x_k)}{(x_k-x_i)(x_j-x_i)}+\frac{\alpha_j(x_{ij}-x_j)}{(x_{ij}-x_i)(x_j-x_i)}.
\end{equation*}
Multiplying this by the octahedron relation,
\[
\frac{(x_{ik}-x_i)(x_{jk}-x_k)}{(x_{ik}-x_k)(x_{jk}-x_j)} =
\frac{(x_{ij}-x_i)}{(x_{ij}-x_j)},
\]
one easily puts the result into the form
\begin{equation*}\label{eq: Q1 Ej trans}
    \frac{\alpha_k(x_{jk}-x_k)}{(x_{jk}-x_j)(x_k-x_j)} =  \frac{\alpha_i(x_{ij}-x_i)}{(x_{ij}-x_j)(x_i-x_j)}+\frac{\alpha_j(x_i-x_k)}{(x_k-x_j)(x_i-x_j)},
\end{equation*}
which, in turn, is easily transformable to equation
\begin{equation*}\label{eq: Q1 Ej}
 \frac{\alpha_k}{x_{jk}-x_j}-\frac{\alpha_i}{x_{ij}-x_j}-\frac{\alpha_k-\alpha_j}{x_k-x_j}+
 \frac{\alpha_i-\alpha_j}{x_i-x_j}  =  0,
\end{equation*}
which is corner equation $(E_j)$.
\end{proof}

\section{System of corner equations for a three-point 2-form not corresponding to quad-equations}
\label{sect: pluri non-quad}

In this section, we show that the general three-point ansatz (\ref{lag LN ij}) is not necessarily related to integrable quad-equations. Consider the Lagrangian 2-form (\ref{lag LN ij}) with the functions $L$ and $\Lambda$ defined by
\begin{equation}\label{eq: ex L}
\frac{\partial}{\partial x} L(x,y;\alpha)=\log\left(\alpha-e^{y-x}\right)=
-\frac{\partial}{\partial y}L(x,y;\alpha)
\end{equation}
and
\begin{equation}\label{eq: ex Lambda}
\frac{\partial}{\partial x}\Lambda(x,y;\alpha,\beta)=
\log\frac{\alpha-\beta e^{y-x}}{\beta-\alpha e^{y-x}}=
-\frac{\partial}{\partial y}\Lambda(x,y;\alpha,\beta).
\end{equation}

We mention that the discrete time relativistic Toda system derived from this 2-form would be
\[
\frac{e^{\widetilde{x}_k-x_k}-\alpha}{e^{x_k-\undertilde{x}_k}-\alpha}=
\frac{e^{x_{k+1}-x_k}-\beta}{e^{x_k-x_{k-1}}-\beta}\cdot
\frac{\beta e^{\undertilde{x}_{k+1}-x_k}-\alpha}{\alpha e^{\undertilde{x}_{k+1}-x_k}-\beta}\cdot
\frac{\alpha e^{x_k-\widetilde{x}_{k-1}}-\beta}{\beta e^{x_k-\widetilde{x}_{k-1}}-\alpha}.
\]

The corner equation for this 2-form are conveniently written in terms of the field variable $X=e^x$, and they read:
\begin{equation}\label{eq: ex Ei}\tag{$E_i$}
\frac{\alpha_{j}X_{i}-X_{ij}}{X_i}\cdot
\frac{\alpha_{i}X_{i}-\alpha_{j}X_{j}}{\alpha_{j}X_{i}-\alpha_{i}X_{j}}
\cdot\frac{X_i}{\alpha_{k}X_{i}-X_{ik}}\cdot
\frac{\alpha_{k}X_{i}-\alpha_{i}X_{k}}{\alpha_{i}X_{i}-\alpha_{k}X_{k}}=1,
\end{equation}
and
\begin{equation}\label{eq: ex Eij}\tag{$E_{ij}$}
\frac{\alpha_{j}X_{i}-X_{ij}}{X_{i}}\cdot
\frac{\alpha_{j}X_{ij}-\alpha_{k}X_{ik}}{\alpha_{k}X_{ij}-\alpha_{j}X_{ik}}\cdot
\frac{X_{j}}{\alpha_{i}X_{j}-X_{ij}}\cdot
\frac{\alpha_{k}X_{ij}-\alpha_{i}X_{jk}}{\alpha_{i}X_{ij}-\alpha_{k}X_{jk}}=1.
\end{equation}
\begin{theorem}\label{th: example}
The system of corner equations (\ref{eq: ex Ei}), (\ref{eq: ex Eij}) is consistent. More precisely, any two of the corner equations within one elementary 3D cube are equivalent modulo {\em octahedron relation}
\begin{equation}\label{eq: ex octahedron}
\frac{\alpha_{j}X_{ij}-\alpha_{k}X_{ik}}{X_{i}}+
\frac{\alpha_{k}X_{jk}-\alpha_{i}X_{ij}}{X_{j}}+
\frac{\alpha_{i}X_{ik}-\alpha_{j}X_{jk}}{X_{k}}=0.
\end{equation}
The discrete 2-form $\cL$ is closed on solutions of the system of corner equations.
\end{theorem}
\begin{proof}
Consistency follows from the second claim, regarding the octahedron relation. The latter is verified by a direct computation, as outlined in the proof of Theorem \ref{th: consistency from octahedron}. Observe the appearance of the octahedron equation of the type $(\chi_4)$ in this example! The last statement, regarding the closedness of the 2-form $\cL$, will be proved at the end of the present section.
\end{proof}

\begin{proposition}\label{prop: ex non-ex}
Corner equations (\ref{eq: ex Ei}) and (\ref{eq: ex Eij}) cannot be represented as quotients of the three-leg forms of quad-equations, multi-affine w.r.t. the variables $X$.
\end{proposition}
\begin{proof}
The three-leg form of the quad-equation on the square $(x,x_i,x_{ik},x_k)$ following from (\ref{eq: ex Ei}) would be
\begin{equation}\label{eq: proof aux1}
(\alpha_{k}X_{i}-X_{ik})\cdot
\frac{\alpha_{i}X_{i}-\alpha_{k}X_{k}}{\alpha_{k}X_{i}-\alpha_{i}X_{k}}\cdot \psi_1(X_i,X)=1,
\end{equation}
while the three-leg form of the equation on the square $(x_j,x_{ij},x_{ijk},x_{jk})$ following from (\ref{eq: ex Eij}) would be
\[
\frac{\alpha_{i}X_{j}-X_{ij}}{X_{j}}\cdot
\frac{\alpha_{i}X_{ij}-\alpha_{k}X_{jk}}{\alpha_{k}X_{ij}-\alpha_{i}X_{jk}}\cdot \psi_2(X_{ij},X_{ijk})=1.
\]
The latter equation, downshifted in the coordinate direction $j$, reads
\begin{equation}\label{eq: proof aux2}
\frac{\alpha_{i}X-X_{i}}{X}\cdot
\frac{\alpha_{i}X_{i}-\alpha_{k}X_{k}}{\alpha_{k}X_{i}-\alpha_{i}X_{k}}\cdot \psi_2(X_{i},X_{ik})=1.
\end{equation}
Comparing (\ref{eq: proof aux1}) and (\ref{eq: proof aux2}), we see that
\[
\frac{\psi_1(X_i,X)X}{\alpha_iX-X_i}=\frac{\psi_2(X_i,X_{ik})}{\alpha_kX_i-X_{ik}}.
\]
Thus, both sides of this identity should be independent of $X$ and of $X_{ik}$, i.e., both should be equal to some $\psi(X_i)$. Thus, we arrive at the following hypothetic form of the quad-equation on the square $(x,x_i,x_{ik},x_k)$:
\[
(\alpha_kX_i-X_{ik})(\alpha_iX_i-\alpha_kX_k)(\alpha_i X-X_i)=(\alpha_kX_i-\alpha_iX_k)X\psi(X_i).
\]
However, there is no choice of $\psi(X_i)$ which would make this to an affine polynomial in $X_i$.
\end{proof}

An explanation of this surprising non-existence result is as follows. Consider a three-point Lagrangian 2-form (\ref{lag LN ij}) with the function $\Lambda$ given by (\ref{eq: ex Lambda}) and the function $L$ being a regular one-parameter perturbation of (\ref{eq: ex L}) defined by
\[
\frac{\partial}{\partial X}L\left(x,y;\alpha\right)=
\log\frac{\alpha-e^{y-x}}{1-\gamma\alpha e^{y-x}}=-\frac{\partial}{\partial y}L\left(x,y;\alpha\right).
\]
Then also the corner equations become regular perturbations of (\ref{eq: ex Ei}) and (\ref{eq: ex Eij}): the first turns into
\begin{equation}\tag{$E_i$}
\frac{\alpha_{j}X_{i}-X_{ij}}{X_{i}-\gamma\alpha_{j}X_{ij}}\cdot
\frac{\alpha_{i}X_{i}-\alpha_{j}X_{j}}{\alpha_{j}X_{i}-\alpha_{i}X_{j}}\cdot
\frac{X_{i}-\gamma\alpha_{k}X_{ik}}{\alpha_{k}X_{i}-X_{ik}}\cdot
\frac{\alpha_{k}X_{i}-\alpha_{i}X_{k}}{\alpha_{i}X_{i}-\alpha_{k}X_{k}}=1,
\end{equation}
while the second turns into
\begin{equation}\tag{$E_{ij}$} 
\frac{\alpha_{j}X_{i}-X_{ij}}{X_{i}-\gamma\alpha_{j}X_{ij}}\cdot
\frac{\alpha_{j}X_{ij}-\alpha_{k}X_{ik}}{\alpha_{k}X_{ij}-\alpha_{j}X_{ik}}\cdot
\frac{X_{j}-\gamma\alpha_{i}X_{ij}}{\alpha_{i}X_{j}-X_{ij}}\cdot
\frac{\alpha_{k}X_{ij}-\alpha_{i}X_{jk}}{\alpha_{i}X_{ij}-\alpha_{k}X_{jk}}=1.
\end{equation}
These equations {\em can} be written as quotients of three-leg forms of multi-affine quad-equations:
\begin{multline}
 \left(\gamma\cdot\frac{\alpha_{j}X_{i}-X_{ij}}{X_{i}-\gamma\alpha_{j}X_{ij}}\cdot
\frac{\alpha_{i}X_{i}-\alpha_{j}X_{j}}{\alpha_{j}X_{i}-\alpha_{i}X_{j}}\cdot
\frac{\alpha_{i}X-X_{i}}{X-\gamma \alpha_{i}X_{i}}\right)  \nonumber\\
\cdot\left(\gamma\cdot\frac{\alpha_{k}X_{i}-X_{ik}}{X_{i}-\gamma\alpha_{k}X_{ik}}\cdot
\frac{\alpha_{i}X_{i}-\alpha_{k}X_{k}}{\alpha_{k}X_{i}-\alpha_{i}X_{k}}\cdot
\frac{\alpha_{i}X-X_{i}}{X-\gamma \alpha_{i}X_{i}}\right)^{-1}=1, \qquad\tag{$E_i$}
\end{multline}
and
\begin{multline}
 \left(\gamma\cdot\frac{\alpha_{j}X_{i}-X_{ij}}{X_{i}-\gamma\alpha_{j}X_{ij}}\cdot
\frac{\alpha_{j}X_{ij}-\alpha_{k}X_{ik}}{\alpha_{k}X_{ij}-\alpha_{j}X_{ik}}\cdot
\frac{\alpha_{k}X_{ij}-X_{ijk}}{X_{ij}-\gamma\alpha_{k}X_{ijk}}\right)\\
\cdot\left(\gamma\cdot\frac{\alpha_{i}X_{j}-X_{ij}}{X_{j}-\gamma\alpha_{i}X_{ij}}\cdot
\frac{\alpha_{i}X_{ij}-\alpha_{k}X_{jk}}{\alpha_{k}X_{ij}-\alpha_{i}X_{jk}}\cdot
\frac{\alpha_{k}X_{ij}-X_{ijk}}{X_{ij}-\gamma\alpha_{k}X_{ijk}}
\right)^{-1}=1. \qquad \tag{$E_{ij}$}
\end{multline}
Indeed, equation
\begin{equation}\label{eq: ex 3leg yes}
\gamma\cdot\frac{\alpha_{k}X_{i}-X_{ik}}{X_{i}-\gamma\alpha_{k}X_{ik}}\cdot
\frac{\alpha_{i}X_{i}-\alpha_{k}X_{k}}{\alpha_{k}X_{i}-\alpha_{i}X_{k}}\cdot
\frac{\alpha_{i}X-X_{i}}{X-\gamma \alpha_{i}X_{i}}=1
\end{equation}
is equivalent to the multi-affine quad-equation
\begin{multline}\label{eq: ex Q3}
\alpha_{k}\left(1-\gamma\alpha_{i}^{2}\right)\left(XX_{i}+\gamma X_{k}X_{ik}\right)-
\alpha_{i}\left(1-\gamma\alpha_{k}^{2}\right)\left(XX_{k}+\gamma X_{i}X_{ik}\right)\\
-\gamma\left(\alpha_{k}^{2}-\alpha_{i}^{2}\right)\left(XX_{ik}+X_{i}X_{k}\right)=0.
\end{multline}

Equation (\ref{eq: ex Q3}) is an obviously re-scaled form of Q3$_{\delta=0}$, so that the closedness of the 2-form $\cL$ for the $\gamma$-deformed system of corner equations follows from the results of \cite{LN1, BS2}, as described in the proof of Theorem \ref{Th: pluri quad consist}. Since the $\gamma$-deformation of both the 2-form $\cL$ and the corner equations is regular, the closedness of the two-form stated in Theorem \ref{th: example} follows.

On the other hand, the three-leg form (\ref{eq: ex 3leg yes}) is the only object in this construction whose $\gamma$-deformation is singular, which explains the non-existence of the three-leg forms in Proposition \ref{prop: ex non-ex}.

As the last remark in this section, we mention that consistency of $\gamma$-deformed corner equations is governed by the following octahedron relation:
\[
\frac{X_{i}-\gamma\alpha_{j}X_{ij}}{X_{j}-\gamma\alpha_{i}X_{ij}}\cdot
\frac{X_{j}-\gamma\alpha_{k}X_{jk}}{X_{k}-\gamma\alpha_{j}X_{jk}}\cdot
\frac{X_{k}-\gamma\alpha_{i}X_{ik}}{X_{i}-\gamma\alpha_{k}X_{ik}}=1.
\]
Also for this element of the construction, the $\gamma$-deformation is regular, and the limit $\gamma\to 0$ leads to the octahedron relation (\ref{eq: ex octahedron}).

\section{Conclusions}

In this paper, we formulated the notion of consistent pluri-Lagrangian systems, which, in our view, should be treated as an answer to the question posed in the title. However, much research is still to be done to justify this proposal. In particular, the relation to more common notions of integrability has only been demonstrated for one-dimensional systems \cite{S}. We expect similar results for two-dimensional systems, as well.

The new definition seems to be capable of being put at the basis of the classification task. Classifying all consistent systems of corner equations (within a certain ansatz for the discrete 2-form) is therefore an extremely important problem.

It would be also very useful to elaborate on the precise relations of the new notion with its predecessors mentioned in the introduction (theory of pluriharmonic functions; Z-invariant models of statistical mechanics and their quasiclassical limit; variational symmetries). All this will be the subject of our ongoing research.
\medskip

This research is supported by the DFG Collaborative Research Center TRR 109 ``Discretization in Geometry and Dynamics''.


\bibliographystyle{amsalpha}

\end{document}